\documentclass{article}
\usepackage{amssymb,amsmath,amsfonts,graphicx,amsthm,rotating,url}
\usepackage{bbm}
\usepackage[usenames,dvipsnames]{xcolor}
\usepackage{wrapfig}
\usepackage[all]{xy}
\usepackage{graphicx}
\usepackage{wrapfig}
\SelectTips{cm}{10}
\usepackage{listings}
\usepackage{courier} 

\newtheorem{theorem}{Theorem}[section]

\newtheorem{lemma}[theorem]{Lemma}
\newtheorem{definition}[theorem]{Definition}
\newtheorem{proposition}[theorem]{Proposition}
\DeclareMathOperator{\lowest}{lowest}

\newcommand{\Z}{\mathbbm{Z}}
\newcommand{\Q}{\mathbbm{Q}}
\newcommand{\R}{\mathbbm{R}}

\newcommand\makevec[1]{{\boldsymbol #1}}
\def \xx {\makevec{x}}

\def \aa {\makevec{a}}
\def \mm {\makevec{m}}
\def \gg {\makevec{g}}
\def \nn {\makevec{n}}

\def\:{\colon}

\def\indef#1{\emph{#1}}
\newcommand{\heading}[1]{\vspace{1ex}\par\noindent{\bf\boldmath #1}}

\newcommand\framedpar[1]{\begin{center}\framebox{
\begin{minipage}{0.9\linewidth}
\vspace{1mm}#1 \vspace{1mm}
\end{minipage}~}\end{center}}

\usepackage{ifthen}
\newboolean{Comments}
\setboolean{Comments}{true}

\newcommand{\marek}[1]{\ifhmode\newline\fi
{\color{Orange}\textsf{*** (MAREK: ) #1\\}}}

\newcommand{\peter}[1]{\ifhmode\newline\fi
{\color{blue}\textsf{*** (PETER: ) #1\\}}}

\begin{document}

\title{Solving Equations and Optimization Problems with Uncertainty\thanks{The research 
of Peter Franek received funding from Austrian Science Fund (FWF): M 1980 and
from the Czech Science Foundation (GACR) grant number 15-14484S with institutional support RVO:67985807.
The research of Marek Kr\v{c}\'al was supported by the Seventh Framework Programme (291734).}
}


\author{Peter Franek \and  Marek Kr\v c\'al \and Hubert Wagner}

\maketitle

\begin{abstract}
We study the problem of detecting zeros of continuous functions that are known only up to an error bound, extending
the theoretical work of~\cite{nondec} with explicit algorithms and experiments with an implementation.\footnote{\url{https://bitbucket.org/robsatteam/rob-sat}}
Further, we show how to use the algorithm for approximating worst-case optima in optimization problems in which 
the feasible domain is defined by the zero set of a function $f: X\to\R^n$ which is only known approximately.

The algorithm first identifies a~subdomain $A$ where the function $f$ is provably non-zero, 
a~simplicial approximation $f': A\to S^{n-1}$ of $f/|f|$, and then verifies non-extendability of $f'$ to $X$ to certify a~zero.
Deciding extendability is based on computing the cohomological
\emph{obstructions} and their persistence. We describe an explicit algorithm for the \emph{primary and secondary 
obstruction}, two stages of a sequence of algorithms with increasing complexity. Using elements and techniques of persistent homology,
we quantify the persitence of these obstructions and hence of the robustness of zero.

We provide experimental evidence that for random Gaussian fields,
the \emph{primary obstruction}---a much less computationally demanding test than the secondary obstruction---is typically 
sufficient for approximating robustness of zero.
\end{abstract}
\section{Introduction}
\label{s:intro} 
\heading{Motivation.}
Detecting zeros of $\R^n$-valued functions is  equivalent to solving systems of real equations, a fundamental problem
of mathematics and theoretical computer science. Our research is  motivated by practical applications, in which the data is often known only approximately. 
We address the case where the input data is limited to the  \emph{approximate} values of a~continuous function $f$ with
values in $\R^n$, sampled over a finite point set.
This uncertainty  is handled in a deterministic way: we aim at verifying that \emph{each} continuous
function compatible with our partial knowledge of $f$ has a~zero. 

Functions that are known only approximately appear in various contexts and are handled in different ways. 
For example, rounding errors in floating-point computations are systematically treated by methods of \emph{interval}
\begin{wrapfigure}{r}{4cm}
\includegraphics{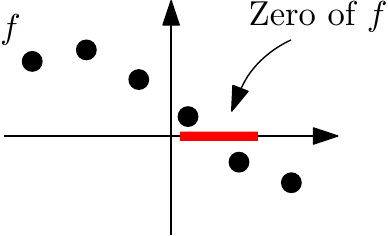}
\caption{For scalar valued function, existence of a~zero can be verified via the intermediate value theorem.}
\label{f:fig1}
\end{wrapfigure}
\emph{arithmetic} and detection of zeros 
resistant to bounded errors is a~frequent problem in this field \cite{Frommer:05,Dian:03,Alefeld:04,Jaulin1,Jaulin4}. 
Other instances of uncertain functions come from measurements of physical quantities, such as in medical 
imaging~\cite{interlevel,gao2013segmenting,Chung2009} 
or robotics~\cite{Merlet:2009,Jaulin3}.
We suppose that potential applications include robust detection of level sets $f^{-1}(a)$ in medical image processing, 
analysing robot trajectory based on data obtained from sensors~\cite{Jaulin3,Jaulin4},
or computing the \emph{inner approximation} of reachable regions of a robotic arm~\cite{Jaulin5}.
The algorithm could also be exploited for analysis of functions obtained by regression (say, in machine learning),
where the function is chosen to fit some given set of sampled values.

To verify that a function $f$, of which we only have a~limited knowledge, has a zero, is equivalent to showing that \emph{each} 
potential candidate $g$ for $f$ has a zero.  If we only have access to sampled values of $f$ and a~Lipschitz constant,
then the set of all such admissible functions $g$ is huge and can not be finitely parametrized. 
However, methods of computational homotopy theory can be applied: 
the closely related problem of verifying that each continuous 
$r$-perturbation of a~given function has a~zero, can be reduced to the topological extension problem for maps into a sphere~\cite{nondec}. 
The latter problem can be addressed via means of obstruction theory, using an algorithmic construction of Postnikov towers.
Such construction has never been implemented and is in its full generality probably out of reach, given the limitations of
computer power. Using a number of simplifications as well as some methods of persistent homology, we present
a~partial solution to the above problem accompanied by an implementation, complexity analysis and several computational experiments.

\heading{Statement of the results.} 
We present an algorithm for detecting zeros of vector valued functions $f:X\to\R^n$ on a compact space $X$ and for approximating the 
\emph{robustness} of zero, that is, a~maximal real number $r>0$ such that every continuous $g: X\to\R^n$ satisfying $\|g-f\|\leq r$
has a zero. By $\|f\|$ we denote the max-norm $\max_{x\in X} |f(x)|$ where $|\cdot|$ is a fixed $\ell_p$ norm in $\R^n$. 
Nontrivial cases happen if $\dim X\geq n$, as otherwise arbitrarily small perturbations of $f$ avoid zero. 
For computer representation we assume that the space $X$ is a simplicial complex.
The map $f\:X\to\R^n$ is specified by its values on the vertices which are assumed to be rational,
and by a rational value $\alpha>0$ such that $|f(x)-f(y)|\le \alpha$ for arbitrary points $x$ and $y$ of any simplex of $X$. 
We emphasize that the precise knowledge of $f$ is not needed.
The algorithm computes a number $r_1\in\R$ such that
\begin{itemize}
\item Every continuous $g$, $\|g-f\| \leq  r_1$, has a zero.
\end{itemize}
A positive $r_1>0$ is then a~certificate of existence of zero of $f$: 
we will say that $f$ has an \emph{$r_1$-robust zero}.
Otherwise the algorithm outputs a negative number and gives no guarantee of the existence of zero.
Under the dimensional constraints $\dim X\leq n+1$ or $n<3$, it also computes a number $r_2>r_1$ such that
\begin{itemize}
\item Some continuous $g$, $\|g-f\|\leq r_2$, has no zero.
\end{itemize}
Under this dimensional constraint, the gap $r_2-r_1$ provably converges to zero, 
if the constant $\alpha$ (and hence our lack of knowledge of $f$) goes to zero. 

\begin{figure}
\begin{center}
\includegraphics{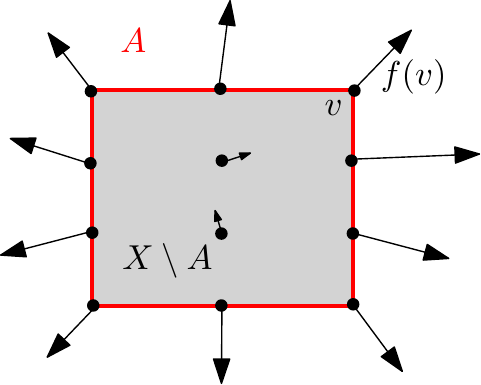}
\end{center}
\caption{This illustrate a~function $f: X \to\R^2$ such that $X=[-1,1]^2$, $A=\partial X$,  and 
$f: A\to \R^2\setminus \{0\}$ is homotopic to the identity map. Then the non-extendability of $f|_A$ to $X\to \R^2\setminus \{0\}$
implies the existence of a~zero in $X\setminus A$. We can bound the robustness of zero in $X\setminus A$ by $\min_{v\in A} |f(v)|$ from below.}
\end{figure}

The main step is to find a subdomain $A\subseteq X$ where $f$ is provable nonzero and where our knowledge of $f$ is sufficient
to determine the homotopy class of $f|_A$ as a~map to $\R^n\setminus \{0\}\simeq S^{n-1}$.
Then non-extendability to $X\to S^{n-1}$ is a certificate of zero. The primary obstruction measures non-extendability to the $n$-skeleton
of $X$ and the secondary obstruction the non-extendability to the $(n+1)$-skeleton. The constraint $\dim X\leq n+1$ could be generalized, 
if we implemented oracles for computing higher obstructions, such as discussed in~\cite{nondec,polypost}.

Our second result is based on computational experiments with random functions. A natural informal question is
\begin{itemize}
\item {\it
How typical are functions for which higher obstructions are needed for detecting a zero?}
\end{itemize}
An example of a~function with nontrivial secondary obstruction is any map $f$ from a $4$-ball $B^4$ to $\R^3$ such that $f|_{\partial B^4}$
is homotopic to the Hopf map $S^3\to S^2$ and hence cannot be extended to $B^4\to \R^3\setminus \{0\}$.
Such property can be verified, if we are given a~sample of function values and a Lipschitz 
constant. Moreover, the homotopy class of $f|_{\partial B^4}$ does not change, if we slightly perturb these sampled function values.

Surprisingly, when performing experiments with random functions (mainly random Gaussian fields), we observed that higher obstructions
are typically \emph{not} needed. 
Whenever we detected a zero of a~randomly generated uncertain function, it was via means of primary obstruction only.
We performed experiments with various random functions from a~triangulated $4$-cube or from a~$4$-torus 
into $\R^3$ as well as from a $5$-cube resp. $5$-torus into $\R^4$.
This observation---if confirmed by theory or by more experiments in different 
settings---could justify the usage of \emph{only the primary obstruction} in potential 
future engineering applications. 

\heading{State of the art.}
Algorithms for detecting zeros used in software packages are based on iterative methods which are often applicable 
if $f$ is given by formulas and is differentiable.
However, these algorithms usually give no guarantees of correctness: the satisfiability of $f(x)=0$ is undecidable for any
class of real functions $f$ that contain polynomials and the sine function~\cite{Wang:74}.

A number of methods has been proposed for testing the (non-)existence of zeros of continuous functions,
exploiting tools ranging from iterative methods in numerical analysis to topology.
The problem has been most studied in the case $\dim X=n$. If  $B^n$ is a unit ball in $\R^n$, then
verifying zeros of $f: B^n\to\R^n$ is equivalent to verifying a fixed point of $f+\mathrm{id}$: here the Brouwer fixed point theorem can be applied~\cite{Rump:2010}.
Other methods for zero verification were studied in the field of interval arithmetic, such as \emph{Miranda's test}~\cite{Alefeld:01},
\emph{Borsuk's test}~\cite{Frommer:05} and the \emph{degree test}~\cite{Franek-al}.
All of these tests have topological flavour and are stable with respect to perturbations of the input function. 
It is shown in~\cite{Franek_Ratschan:2015}  that the \emph{degree test} can detect a zero of $f: B^n\to\R^n$ whenever
the zero of $f$ is \emph{robust} (that is, each $g$ close enough to $f$ has a zero). 
The above mentioned primary obstruction directly reduces to the degree test if $\dim X=n$.
While the topological degree computation has been explicitly described in the literature and also implemented~\cite{degpaper}, 
the problem is far more complicated if the domain $X$ has larger dimension than $n$.
In~\cite{nondec} we showed that the existence of a~robust zero of piecewise linear functions $X\to\R^n$ is undecidable 
if $n\geq 3$ is a fixed odd integer and $X$ is a~$(2n-2)$-dimensional simplicial complex 
($X$ is considered to be a~part of the input).

Zero sets of functions with inherent uncertainty have been studied via means of computational topology 
in the context of~\emph{well groups}~\cite{ComputeWell}.
In their general settings, well groups associated to $f: X\to Y$ and a subspace $Y'\subseteq Y$ 
describe properties of the preimage $f^{-1}(Y')$ which persist if we perturb the input function $f$. 
In the important case of $Y=\R^n$ and $Y'=\{0\}$, well groups describe zero sets of functions: namely, the zeroth well
groups measures robustness of \emph{existence} of zero and higher well groups reflect further topological properties of the zero sets.
In~\cite{well-socg} we showed that the primary obstruction can be used to compute a~certain subgroup of the well group
that in many cases coincides with the full well group (see~\cite[Thm. 1.4]{well-socg}).
A general algorithm for well group computation cannot be expected because the above mentioned undecidability result~\cite{nondec}
directly transfers  to well groups when $\dim X\ge 2n-2$ as well. 
Our implementation can be thought of as~an approximation of the zeroth well group that extends the work of~\cite{chazal} 
where the special case $\dim X=n$ is solved. 

An application of our algorithm is in worst-case analysis of optimization problems where the feasible domain is defined by equations.
The worst-case approach in robust optimization has been widely studied, see~\cite{Ben-Tal2002,Ben:2009,bertsimas2011theory,Beyer20073190}. 
Usually, the uncertainty applies to a finite number of parameters which are assumed to be taken from a~known domain. 
In our approach, we rather work with the space of \emph{all continuous functions} that are compatible
with our partial knowledge of $f$. 

\heading{Outline and organization of the paper.}
In the algorithm, we first create a~filtration $\{A_r\supseteq A_s\}_{r\leq s}$ of subcomplexes that ``approximate'' the topological spaces 
$A(r):=\{x\in X:\,\,|f(x)|\geq r\}$. 
We compute a~simplicial approximation $f': A_r\to \Sigma$ of $f$ where $\Sigma$ is a given triangulation
of the $(n-1)$-sphere.
Then we ask for the smallest $r$ such that the restriction of $f'$ to $A_r$ can be extended to
all of $X$, and show that the robustness of zero of the original function is $\alpha$-far from $r$.

Such extendability is decidable if $\dim X\leq 2n-3$~\cite{ext-hard}, but the only procedere for this we are aware of
is based on the~algorithm for computing stages of Postnikov towers from~\cite{polypost} 
that depends on several other papers~\cite{pKZ1,post,vokrinek:oddspheres} and is unlikely to be fully implemented in near future. 
Instead of that, we implemented a~persistent version of both the \emph{primary} and \emph{secondary obstruction}, 
which test extendability to the $n$- and $(n+1)$-skeleton of $X$.\footnote{The only exception is the case $n=3$, $\dim X>3$ where
the triviality of secondary obstruction is undecidable in general. However, if $X$ is assumed to be a triangulation of the cube $[0,1]^4$,
then our algorithm works with no essential changes. For many other fixed $4$-dimensional spaces $X$ the problem is decidable too.} 
First, we compute the maximal $r_1$ for which the cohomological obstructions to extending 
${f'}|_{A_{r_1}}$ to $A_{r_1}\cup X^{(n)}$ (\emph{primary obstruction}) does not vanish.
Similarly, we compute a~maximal $r_2\geq r_1$ for which $f'|_{A_{r_2}}$ is not extendable to $A_{r_2}\cup X^{(n+1)}$ 
(non-vanishing of the \emph{secondary obstruction}): this requires us to parametrize all extensions to the $n$-skeleton.

In Section~\ref{s:discretize} we show how to approximate the spaces $A(r)$ by simplicial complexes 
and the sphere-valued map $f/|f|$ via a~simplicial map. 
A~high-level description of our algorithm is in Sections~\ref{sec:algorithm-oracle} and \ref{s:obstructions},
with a partial~lower-level description in Appendix~\ref{s:secondary-persistence} and \ref{s:implementation}.
In Section~\ref{s:rob-opt} we show how to use the method for approximating
worst-case optima in optimization problems where the feasible domain is defined by equations.
In Section~\ref{s:experiment} we present some computational experiments with random Gaussian fields.
More details about testing and performance are delegated to Appendix~\ref{s:formulas}. 
The last section contains theoretical worst-case complexity bounds.

\section{Discretizing the function $f$}
\label{s:discretize}
In this section we show how to convert the ``unknown'' continuous function $f$ to its discrete simplicial approximation.
\begin{definition}
\label{d:filtration}
A \indef{continuous filtration} of spaces is a family $(A_r)_{r\in\R}$ such that $A_r\supseteq A_s$ whenever $r\le s$.
A continuous filtration $(A_r)_{r\in\R}$ is called \indef{step-like} whenever there exists a sequence of numbers $-\infty=:r_{-1}<r_0\le r_1\le r_2\le\ldots\le r_k$ such that for any $r,s\in (r_i,r_{i+1}]$, $A_r=A_s$ holds for all $i$. 
Continuous filtrations $(A_r)_r$ and $(B_r)_r$ are called \indef{$\alpha$-interleaved} whenever $B_{r+\alpha}\subseteq A_r$ and $A_{r+\alpha}\subseteq B_r$ for each $r\in\R$.
\end{definition}

\begin{definition}
\label{d:A-r} 
Let $f\:X\to \R^n$ be a continuous map on a simplicial complex $X$ and let $|\cdot|$ be a norm on $\R^n$. 
\begin{enumerate} 
\item By $A_r$ we denote the \emph{subcomplex} of $X$ spanned by the vertices $v$
of $X$ with $|f(v)|\ge r$. 
\item By $A(r)$ we denote the \emph{subspace} of $X$ defined by $A(r)=\{x\in X\: |f(x)|\ge r\}$. 
\item We say that $f$ is \indef{simplexwise $\alpha$-Lipschitz} whenever $|f(x)-f(y)|\le \alpha$ for each pair of points $x,y\in\Delta$ of any simplex $\Delta\in X$.
\end{enumerate}
\end{definition}

The spaces $A_r$ form a step-like filtration where a step occurs for each $r$ equal to $|f(v)|$ for some vertex $v$ of $X$.

Let $e_j=(0,\ldots,1,0,\ldots,0)$, $j=1,\ldots,n$ be the unit vectors in $\R^n$ in the direction of the axes and let
$\Sigma^{n-1}$ be the simplicial model of the $(n-1)$-sphere obtained from the boundary of a cross-polytope. 
More explicitely, the vertex set of $\Sigma^{n-1}$ is  $\{\pm e_j\,|\,j=1,\ldots, n\}$
and the triangulation consists of all simplices spanned by vertex sets that do not contain any antipodal pair.
A natural sphere-valued approximation of $f$ is then given by the map $f'$ as follows.
\begin{definition}
\label{d:vertex-app}
Let $f: X\to\R^n$ and $V$ be a subset of the vertices of $X$. We define the \indef{vertex approximation} 
$f': V\to \{e_1, -e_1,\ldots, e_n, -e_n\}$ to be the map that to a vertex $v$ assigns $s_j e_j$, where
$j$ is the index of the component of $f(v)$ with \emph{largest} absolute value and $s_j$ is the sign of $f_j(v)$.\footnote{For example,
if $f(v)=[2,-3]$, then we choose $f'(v)=-e_2$. If there are more components of $f(v)$ with the same absolute value, we choose one by an~arbitrarily chosen rule.} 
\end{definition}
\begin{lemma}\label{t:approx}
Let $f\:X\to \R^n$ be a simplexwise $\alpha$-Lipschitz map for some constant $\alpha>0$ and $A_r$, $A(r)$ be the filtrations from Definition~\ref{d:A-r},
defined with respect to the $\ell_p$-norm for some $p\in [1,\infty]$.

Then the following holds: 
\begin{enumerate}
\item The continuous filtrations $(A_r)_{r\in \R}$ and \(A(r)_{r\in\R}\) are $\alpha$-interleaved. 
\item If $r> \alpha n^{1/p}/2,$ the vertex approximation $f': V(A_r)\to V(\Sigma^{n-1})$
defines a simplicial map $f'\:A_r\to\Sigma^{n-1}$ (that is, it maps simplices to simplices). 
\item If $r>\alpha n^{1/p}$, then $f'\:A_r\to \Sigma^{n-1}\subseteq \R^n\setminus\{0\}$ 
is homotopic to $f|_{A_r}\:A_r\to \R^n\setminus\{0\}$.
\end{enumerate}
\end{lemma}

The simplicial map $f'\:A_r\to\Sigma^{n-1}$ as above will be called the \indef{simplicial approximation of $f|_{A_r}$.}
We note that the choice of the sphere model and the discretization $f'$ of $f$ is independent of the rest of the algorithm given in the following chapters
and is not the only possible choice: however, we couldn't find one with approximative properties better than in Lemma~\ref{t:approx}.
\begin{proof}[Proof of part 1]
If a~point $x\in |X|$ is not in $A_r$, we know that  $|f(v)|<r$ for some vertex $v$ of a~simplex $\Delta$ supporting $x$. 
The simplexwise $\alpha$-Lipschitz property then implies that $|f(x)|<r+\alpha$, hence $x\notin A(r+\alpha)$. This proves
$A(r+\alpha)\subseteq A_r$.

The other inclusion holds because once a point $x$ of a simplex $\Delta\in X$ is in a~given $A_r$, 
then for arbitrary vertex $v$ of $\Delta$ holds $|f(v)|\ge r$ and  $|f(x)-f(v)|\le \alpha$. 
Thus $|f(x)|\ge r-\alpha$, hence $x\in A(r-\alpha)$. 
\end{proof}
\begin{proof}[Proof of part 2]
We want to prove that no adjacent vertices $u$ and $v$ are mapped by $f'$ to $e_i,-e_i$ for some $i$. 
Without loss of generality we can assume that $f'(v)=e_1$. Assuming $|f(v)|\geq r$ and the definition of $f'$, 
the first component (with largest absolute value) must satisfy $f_1(v)\ge rn^{-1/p}$.
Similarly if $f'(u)=-e_1,$ then $f(u)_1\le-rn^{-1/p}$, but this would imply that $\alpha\geq |f(u)-f(v)|\ge2rn^{-1/p}$, 
contradicting the assumption $r>\alpha n^{1/p}/2$.
\end{proof}
\begin{proof}[Proof of part 3]
We show that the simplexwise straight-line homotopy between $f$ and $f'$ has values in $\R^n\setminus \{0\}$. 
Let $\Delta\in X$, $v$ be a vertex of $\Delta$, $x\in\Delta$, and assume that $r>\alpha n^{1/p}$. Again, assume WLOG that $f'(v)=e_1$.
Then $f_1(x)\geq f_1(v) - \alpha \geq r n^{-1/p}-\alpha >0$ and the straightline homotopy between maps 
$f,f'\:A_r\to\R^n\setminus\{0\}$ has positive first coordinate, hence it avoids zero. 

\end{proof}

\section{The algorithm using an oracle for persistence of obstructions}
\label{sec:algorithm-oracle}
In this section we describe a high-level description of our algorithm for approximating robustness of zero. The specification is as follows:\\

{\bf Input:}
\begin{itemize}
\item $X$, a~simplicial complex,
\item $f: X^{(0)}\to\R^n$, function values at vertices,
\item $\alpha>0$.\footnote{For computability purposes, we assume that $f(v)$ and $\alpha$ are rational or computable.}
\end{itemize}

{\bf Output:}
\begin{itemize}
\item a~lower bound on the robustness of zero (possibly negative),
\item an~upper bound on the robustness of zero (possibly $\infty$).
\end{itemize}

The unknown function $f$ is thus represented by function values in vertices and a simplex-wise Lipschitz constant $\alpha$. 
A negative lower bound or infinite upper bound give no information at all: however, in case of $\dim X\leq n+1$ or $n<3$, the 
lower and upper-bounds on robustness will be at most $2\alpha$-far from each other.
Even outside this dimension range, the computed lower bound on robustness will be at most $\alpha$-far from the robustness of zero 
of $f|_{A\cup X^{(n+1)}}$.

\begin{definition}
\label{d:persistence-obstruction}
Let $X\supseteq A_{r_0}\supseteq A_{r_1}\supseteq \ldots$ be a filtration of simplicial complexes, 
$r_i\leq r_j$ for $i\leq j$ and $f': A_{r_0}\to S^{n-1}$. 
Then the \indef{persistence of primary obstruction} is the largest $r_j$ such that the restriction of $f'$ to $A_{r_j}$ 
\emph{can not} be extended to a~(not necessarily simplicial) map $A_{r_j}\cup X^{(n)}\to S^{n-1}$ where $X^{(n)}$ is the $n$-skeleton of $X$. 
The \indef{persistence of secondary obstruction} is the largest $r_k$
such that the restriction of $f'$ to $A_{r_k}$ can not be extended to $A_{r_k}\cup X^{(n+1)}\to S^{n-1}$.
\end{definition}
In what follows, assume that an oracle is given that, for a filtration of simplicial complexes and a simplicial map $f': A_{r_0}\to\Sigma^{n-1}$, 
computes the persistence of secondary obstrucion. 
We assume that we are given a continuous map $f\:X\to\R^n$ by its values on the vertices of $X$,
its simplexwise Lipschitz constant $\alpha$ and a norm $\ell_p$ on $\R^n$ for $p\in [1,\infty]$.
The outline of the algorithm follows.

{{\renewcommand{\theenumi}{\Alph{enumi}} 
\begin{enumerate}
\item \begin{enumerate}
\item
Label the set 
$$\{|f(v)|\:v\in V(X)\text{ such that } |f(v)|\geq \alpha n^{1/p} \}$$   by $\{r_0,r_1,\ldots,r_h\}$ 
so that $r_i\leq r_j$ for $i\leq j$.
\item 
For any simplex $\Delta\in X$ compute its filtration value $r(\Delta)$ by 
\begin{equation}
\label{e:filtration}
r(\Delta):=\min_{v\text{ vertex of }\Delta} |f(v)|.
\end{equation}
This yields a filtration $A_{r_0}\supseteq\ldots\supseteq A_{r_h}$ that together with the values $r_0,\ldots,r_h$ 
determines the step-like continuous filtration $(A_r)_{r\in\R}$ from Definition~\ref{d:filtration}.
\item 
For vertices $v$ of $X$ with $|f(v)|\ge r_0$ compute the vertex approximation $f'(v)$ via Definition~\ref{d:vertex-app}. 
\end{enumerate}
\item Use the oracle to compute the persistence of secondary obstruction $r_k$ (Def.~\ref{d:persistence-obstruction}). 
\begin{enumerate}
\item {\bf If} $k>0$: output ``robustness of zero is at least  $r_{k}-\alpha.$''\\
      {\bf Else}: output ``no guarantee of zero''
\item {\bf If} $\dim X\leq n+1$ or $n<3$:  output ``robustness of zero is at most $r_k+\alpha$'' \\
      {\bf Else}: output ``no guarantee of upper bound (robustness of zero is at most $\infty$)''
\end{enumerate}
\end{enumerate}}

The constraints in B(b) could be replaced by $\dim X\leq n-1+k$, if we used an oracle for persistence of the first $k$ obstructions.
However, implementing such an~oracle is theoretically possible only if $k<n-1$. In fact, even the special case $\dim X=4$ and $n=3$
is beyond this bound and we cannot implement the oracle for secondary obstruction for this dimension pair
with no restrictions on $X$. In the important special case when $X$ is topologically a cube $[0,1]^4$ and $n=3$, the general algorithm works 
with no essential changes. More details about the implementation of the oracle for this dimension pair
are given in Appendix~\ref{s:secondary-persistence}, p.~\pageref{h:mn=43}.

\begin{theorem}\label{t:main} The above algorithm outputs correct statements. 
\end{theorem}

\begin{proof}
In~\cite[Lemma 3.3]{nondec} we showed that $f$ has an $r$-robust zero iff $f|_{A(r)}$ is not extendable to a nowhere zero function on $X$.

{\it Correctness of B(a)}
Assume that $k>0$ and let $r:=r_k$. Non-extendability of $f'|_{A_r}$ to the $(n+1)$-skeleton $X^{(n+1)}$ implies non-extendability to all of $X$.
By Lemma~\ref{t:approx}, $f'|_{A_r}$ is homotopic to $f|_{A_r}$ and hence non-extendability of the former implies non-extendability of the latter.
Further, the relation $A_r\subseteq A(r-\alpha)$ implies non-extendability of $f|_{A(r-\alpha)}$, which finally implies that $f$ has an $(r-\alpha)$-robust zero on $X$.

{\it Correctness of {B(b)}}
Let $r>r_k$ be arbitrary.
The assumption, the restriction of $f'$ to $A_{r}$ \emph{is} extendable to $A\cup X^{(n+1)}$. 
If $\dim X\leq n+1$, then this is equivalent to the extendability to all of $X$.
The cases $n<3$ reflect low dimensional phenomena: we will show that then the extendability to $A\cup X^{(n)}$ already implies the extendability 
to all of $X$.
If $n=1$, $f'$ has values in the $0$-sphere $S^0\in \{+,-\}$ and if it 
can be extended to the $1$-skeleton, we can assign a sign $+$ or $-$ to each connected component of $X$ and naturally extend to $X\to \{+,-\}$. 
If $n=2$, then $f'$ has values in a circle, $S^1$. Assume that it can be extended to $A\cup X^{(2)}\to S^1$.
Then any extension $A\cup X^{(j)}\to S^1$ of $f'$, $j\geq 2$, can be extended to $A\cup X^{(j+1)}$, because the restriction 
of $g$ to the boundary of any $(j+1)$-simplex $\Delta^{j+1}$ defines a map $\partial \Delta^{j+1}\to S^1$
from a $j$-sphere to the circle and such map is homotopic to a constant (see, e.g.~\cite[Chapter 4.1]{Hatcher}),
hence $\partial \Delta^{j+1}\to S^1$ can always be extended to all of $\Delta^{j+1}$.

Assume that $\dim X\leq n+1$ or $n<3$ and that $r>r_k$. 
Then $r>r_0$ and $f'|_{A_r}$ is well defined and homotopic to $f|_{A_r}$ by Lemma~\ref{t:approx}.
This implies the extendability of $f|_{A_r}$ to all of $X$ 
and the relation $A(r+\alpha)\subseteq A_r$ implies the extendability of $f|_{A(r+\alpha)}$. Thus the robustness is less than $r+\alpha$ for any
$r>r_k$, yielding an upper bound $r_k+\alpha$ on the robustness of zero.
\end{proof}

To conclude this section, we remark that
\begin{itemize}
\item In case when $f$ has no zero at all, we may easily approximate the \emph{robustness of non-existence of zero} by $\min_v |f(v)|-\alpha$.
\item The infinite bound in $B(b)$ can be improved to $\max_v |f(v)|+\alpha$.
\end{itemize}

\section{Persistence of obstructions}
\label{s:obstructions}
In this section we describe the algorithm for computing the persistence of primary obstruction and roughly outline the algorithm for secondary 
obstruction (which is described in more detail in Appendix~\ref{s:secondary-persistence}).

\heading{Primary obstruction---extendability to $X^{(n)}$.} 
Here we review some facts from obstruction theory. 
A reference for the next proposition can be a~textbook such as \cite[III, 1.2]{prasolov}. 
\begin{proposition}[Primary obstruction]\label{p:prim}
Let $A\subseteq X$ be a pair of simplicial complexes, $f\:A\to \Sigma^{n-1}$ be simplicial, $z\in C^{n-1}(\Sigma^{n-1};\Z)$ be a cocycle
generating the cohomology, and $y:=f^ \sharp(z) \in Z^{n-1}(A;\Z)$ its pullback. 

Then $f: A\to \Sigma^{n-1}$ can be extended to a (not necessarily simplicial) map $A\cup X^{(n)}\to\Sigma^{n-1}$, iff 
$y\in Z^{n-1}(A;\Z)$ can be extended to a~cocycle $x\in Z^{n-1}(X;\Z)$ such that $x|_A=y$.
\end{proposition}

Thus extendability of $f$ is reduced to extendability of an $A$-cocycle $y$ to a global cocycle $x$ defined on all of $X$.
We will use the notation $\Omega(A):=\{x\in Z^{n-1}(X;\Z)\: x|_A = y\}$ of all cocycle extensions and want to test its non-emptiness.
$\Omega(A)$ corresponds to solutions of a~linear equation over integers. 
To see that, let $\bar y\in C^{n-1}(X;\Z)$ be an arbitrary cochain (not necessarily a cocycle) such that $\bar y|_A=y$. We have that
\begin{equation}
\label{e:E} 
\Omega(A)=\{\bar y-c\: c\in C^{n-1}(X,A;\Z)\text{ such that }\delta c = \delta \bar y\}.\end{equation} 
Subtracting $c\in C^{n-1}(X,A;\Z)$ does not change the values on $A$-simplices, so any such $\bar y-c$ is still an extension of $y$.
The non-emptiness of $\Omega(A)$ is thus equivalent to solvability of the linear equation  $\delta c=\delta \bar y$ with the 
unknown $c\in C^{n-1}(X,A;\Z)$. 

A natural set of generators of $C^{n-1}(X,A;\Z)$ is the set of all $(n-1)$-simplices in $X$ that are not in $A$
with the identification between a simplex $\Delta$ and its \emph{characteristic cochain} that assigns $1$ to $\Delta$ 
and $0$ to all other simplices.  Converting $\delta c=\delta\bar y$ into an explicit matrix system of linear equations 
then amounts to enumerating the $(n-1)$- and $n$-simplices in $X\setminus A$, computing the codifferential matrix of $\delta$ 
using the definition of boundary and expressing the right-hand side $\delta\bar y$ in the basis of the $n$-simplices.

\heading{Persistence of the primary obstruction---the algorithm.} 
We recall that in the persistent setting the input contains
a~filtration of simplicial complexes $X\supseteq A_{r_0} \supseteq A_{r_1},\ldots,\supseteq A_{r_h}$ 
and a simplicial map $f': A_{r_0}\to \Sigma^{n-1}$.
We want to compute the largest value $j$ such that the restriction of $f'$ to ${A_{r_j}}$ cannot be extended to $A_{r_j}\cup X^{(n)}$.

Let $A:=A_{r_0}$ and $z,y$ be defined as above.
The cochain extension $\bar y$ of $(f')^\sharp(z)$ is also an extension of $(f'|_{A'})^\sharp(z)$ for each $A'\subseteq A$.
Thus we fix one $\bar y$ for all spaces $A_r$.
Then the only thing that is changing in solving $\delta c=\delta \bar y$, $c\in C^{n-1}(X,A_r;\Z)$ with increasing $r$,
is the requirement that $c$ should be zero on $A_r$ and is hence supported on $X\setminus A_r$. Note that
$X\setminus A_r$ becomes larger with increasing $r$: we are allowed to include more columns into our matrix representing $\delta$.

Now we describe the algorithm on a lower level.
\label{page:primary-persistence}
\begin{itemize}
\item
First we choose the cocycle $z\in\Sigma^{n-1}$ that generates the $(n-1)$ cohomology and its pullback
$y:=(f'|_A)^\sharp(z)\in Z^{n-1}(A;\Z)$ for $A=A_{r_0}$. 
\item We fix an arbitrary extension $\bar y\in C^{n-1}(X;\Z)$ of $y$: the simplest option is to choose 
$\bar y(\Delta)=0$ for all $(n-1)$-simplices $\Delta\in X$ that are not in $A$.
\item We compute the filtration values of all $(n-1)$ simplices in $X$ by (\ref{e:filtration}). 
\item We order the $(n-1)$-simplices of $X$ by their filtration value, and choose an arbitrary enumeration of the $n$-simplices.
These choices will serve as bases of the $(n-1)$- and $n$-cochains (we identify simplices and their characteristic cochains).
\item We construct the matrix $M$ representing the codifferential with respect to the bases chosen above. The columns of $M$ are 
coboundaries of the $(n-1)$-simplices ordered by increasing filtration values.
Further, we convert the right-hand side $\delta \bar y$ to an integer~vector $\aa$ using the chosen basis of $n$-simplices.
\end{itemize}
Recall that we want to solve $\delta c=\delta\bar y$ for $c$ that is a linear combination of $(n-1)$-simplices 
with filtration values at most  $r$, where $r$ is as small as possible. 
Such $r$ is then the desired persistence of the primary obstruction: indeed, it is the smallest $r$ such that $\delta\bar y$ 
can be expressed as a~coboundary $\delta c$ where $c$ has filtration at most $r$, but \emph{cannot} be expressed
as $\delta c$ so that $c$ has filtration strictly smaller than $r$.

This directly translates to the following problem, which is the last
step of the persistence-of-primary-obstruction algorithm.

\framedpar{
\label{p:es}
{\bf Problem {\sc Earliest Solution}}
\\
Input: A matrix $M\in \Z^{p\times q}$ and a column vector $\aa\in \Z^p$.
\\
Output: A column vector $\xx\in \Z^q$ such that $M\xx = \aa$.
\\
Objective: Minimize the index of the \emph{last nonzero entry} of $\xx$, that is, $\ell\ge 0$ such that $x_\ell\neq 0$ and $x_{\ell+1}= x_{\ell+2} =\ldots = x_q = 0$.}

The persistence of the primary obstruction is then the filtration value of the $l$-th column. 

The EARLIEST SOLUTION problem could be solved by binary search on the value $\ell$ while solving an~ordinary 
linear system of equations in each iteration. 
Our implementation uses a simple matrix reduction approach (resembling algorithms for persistent homology) 
which avoids the binary search (see Appendix~\ref{s:pers} for details).

\heading{Secondary obstruction---extendability to $A\cup X^{(n+1)}$.}
Computing the secondary obstruction and its persistence contains similar ingredients but is more technical and we postpone a lower-level
description to Appendix~\ref{s:secondary-persistence}. Here we outline the main steps for the non-persistent version with a fixed $A$.
We assume that $f': A\to\Sigma^{n-1}$ is extendable to $A\cup X^{(n)}$ and that $\Omega(A)$ (described by (\ref{e:E})) is nonempty.

We need to implement the ``Steenrod square'' operation on the level of cochains. 
We chose  to use the  notation  from  the original paper of Steenrod~\cite{Steenrod}
$$\smile_{n-3}: C^{n-1}(X;\Z_2)\times C^{n-1}(X;\Z_2)\to C^{n+1}(X;\Z_2)$$
which induces (when an element is ``multiplied'' by itself) the standard operation 
$$Sq^2: H^{n-1}(X;\Z_2)\to H^{n+1}(X;\Z_2)$$ on the level of cohomology for $n>3$ (similarly for relative cohomology).
The algorithm for $\smile_{n-3}$ directly follows from formulas in~\cite[p. 292--293]{Steenrod}. 
For the following facts, we refer to \cite{Steenrod} and \cite{Steenrod:CohomologyOperationsObstructionsExtendingContinuousFunctions-1972}:
\begin{proposition}[Secondary obstruction]
\label{p:secondary}
Let $A\subseteq X$ be a pair of simplicial complexes, $f': A\to\Sigma^{n-1}$ be simplicial and 
assume that the ordering of vertices of $X$ and $\Sigma^{n-1}$ is chosen so that
$v\preceq w\,\,\Rightarrow\,\,f'(v)\preceq f'(w)$.
For each $x\in\Omega(A)$ let $(x\mod 2)$ be the image of $x$ under the natural homomorphism $C^{n-1}(X;\Z)\to C^{n-1}(X;\Z_2)$.  
Then 
\begin{equation}
\label{e:steenrod}
v(x):=(x\mod 2)\smile_{n-3} (x\mod 2)
\end{equation} 
vanishes on $A$, that is, it is an~element of $Z^{n+1}(X,A;\Z_2)$.

Further, if $n>3$, then $f'$ can be extended to a map $X^{(n+1)}\to \Sigma^{n-1}$ iff $v(x)$ 
is a~relative coboundary for some $x\in\Omega(A)$.
\end{proposition}
Thus extendability to $X^{(n+1)}$ is equivalent to satisfiability of the equation $\delta c=v(x)$, $c\in C^n(X,A;\Z_2)$, 
for \emph{some} $x\in\Omega(A)$. To decide this, we parameterize $\Omega(A)$ by a fixed representative $x$ and generators $g_j$ of 
$Z^{n-1}(X,A;\Z)$: an~arbitrary element of $\Omega(A)$ is then $x-\sum_j u_j g_j$ for some $u_j\in\Z$.
To reduce the number of $j$'s, we only need to take generators of the cohomology group $H^{n-1}(X,A;\Z)$.
Exploiting the linearity of the operation $v$  on the level of cohomology (\cite[p. 504]{Steenrod}), we have that 
$v(x-\sum_j u_j g_j)$ is a coboundary iff $v(x)-\sum_j u_j v(g_j)$ is a coboundary. Thus our equation reduces to
$\delta c + \sum_j u_j v(g_j)=v(x)$.
This is a~system of equations with right-hand side $v(x)$ and unknowns $c$ and $u_j$, this time over the $\Z_2$-coefficients.

We also remark that the last proposition is valid also in the case $n=3$ once we replace $\Z_2$-coefficients by $\Z$-coefficients and the $\smile_3$
operation by the cup product. However, deciding whether there exists an $x$ such that  $x\smile x$ is a~coboundary, is hard (and undecidable for general spaces $X$).
We show at the end of Appendix~\ref{s:secondary-persistence} 
that if $X$ is a triangulation of the topological cube $[0,1]^4$ and $n=3$, then
triviality of the secondary obstruction can easily be tested as well: this case is also included in our implementation.

To compute the largest $r$ such that the map $f'|_{A_r}$ is not extendable to $A_r\cup X^{(n+1)}$, we could use a binary search.
As in the case of the primary obstruction, it can be avoided and we can compute the persistence of the secondary obstruction
using a~single matrix reduction: this is explained in Appendix~\ref{s:secondary-persistence}.

\section{Application for robust optimization}
\label{s:rob-opt}
\heading{Reduction of robust optimization to the ROB-SAT problem.}
Our algorithmic approach has a natural extension for  optimization with uncertainty. 
We pose the following optimization problem:
\begin{equation}
\label{e:rob-opt}
\begin{array}{rl}
\text{maximize}& o(x)\\
\text{subject to}& g(x)=0\\
& x\in X
\end{array}
\end{equation}
where $X$ is a compact domain\footnote{Such a~domain implicitly imposes inequality constraints which can be seen as uncertain ones as well if the chosen norm on $\R^n$ is $\ell_\infty$, see \cite{nondec}. Also the function $o$ could be considered as uncertain without adding further complexity to the problem, but we prefer to have the statement as simple as possible.} 
and $g\:X\to\R^n$ is uncertain.
Let as assume, for simplicity, that $r>0$ is fixed and $g$ is an~unknown continuous $r$-perturbation of a~known given map 
$f\:X\to\R^n$. A simple instance of the above problem is visualized below:
\begin{center}
\includegraphics{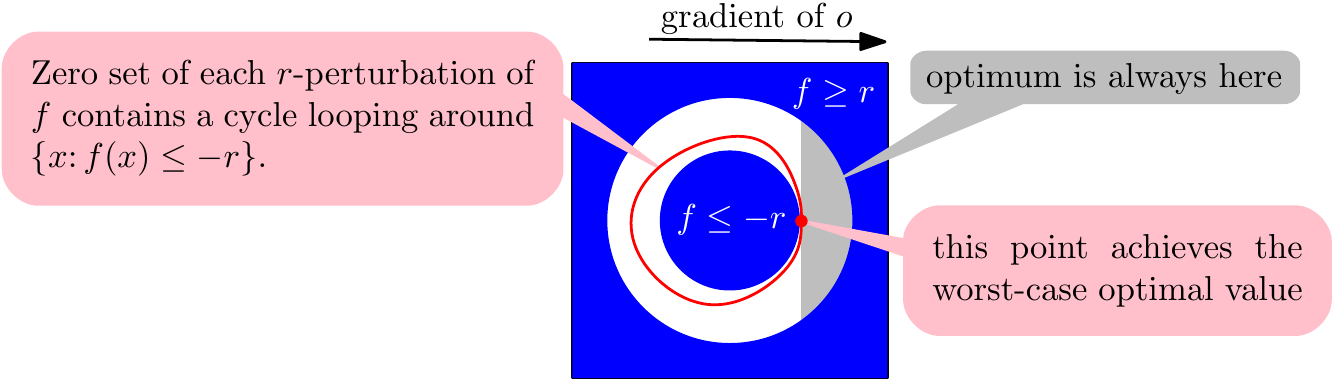}
\end{center}
Both $f$ and $o\:X\to\R$ can be specified in various ways but let us further assume that they are \emph{simplexwise linear} 
and that we know their values on vertices. 

We remark that in a common approach the uncertainty is parametrized, that is, in the problem above, the constraints would have the form $g_p(x)=0$ (or  $g_p(x)\leq 0$), where $p$ is an unknown vector-valued parameter (see~\cite{Ben:2009}). 

The common goal is to compute the optimal value in the worst case, i.e., $$\inf_{\|g-f\|\leq r}\,\max_{x\in g^{-1}(0)} o(x)$$ in our case. The worst-case
optimal value is equal to the maximal number $\beta\in \R$ such that $f$
has an $r$-robust zero on $o^{-1}[\beta,\infty)$. 
If $o$ is simplexwise linear, then $o^{-1}[\beta,\infty)$ can be triangulated and 
the existence of an~$r$-robust zero on $o^{-1}[\beta,\infty)$ can be algorithmically tested via computing higher-order obstructions,
whenever $\dim X\le 2n-3$ or $n<3$~\cite{nondec}. 
The exact worst-case optimal value can be found by doing a binary search on the maximal value 
$\beta$ and using the ROB-SAT algorithm \cite{nondec} in each step.

\heading{Efficient implementation.} Also the efficiency--tuned algorithm presented in this paper can be tweaked into 
the setting of optimization very easily and thus the binary search 
avoided.
We may assume that both $f$ and $o$ are only given via function values in vertices and simplex-wise Lipschitz constants, and want
to approximate $\inf_{\|g-f\|\leq r} \max_{x\in g^{-1}(0)} o(x)$ for some $r$.\footnote{To avoid further simplicial subdivisions, we again need to assume that $r>\alpha n^{1/p}$, i.e., that the description
of $f$ is fine-grained enough for the retrieval of the homotopy
class of $f|_{A_r}$).} 

The only difference occurs before each call of {\sc Earliest Solution} subroutine where we sort the rows of the matrix $M$ and the right-hand side $\aa$ ($n$-simplices in the case of primary obstruction) according to their $o$-filtration value (minimum of $o(v)$ over their vertices $v$). Also we cut off the columns of the matrix with filtration value larger than $r$. After the column matrix reduction as described in Appendix~\ref{s:pers}, the desired approximation of the worst-case optimal value is the $o$-filtration value corresponding
to the row of the lowest nonzero element on the right hand side
$\aa$ after the reduction.

We can immediately compute a lower bound\footnote{An upper bound is obtained when the dimension is at most $n$ or $n\le 2$ for the primary obstruction and at most $n+1$ for the secondary obstruction.} on the \emph{uncertainty-optimality curve }
$\text{OPT}(r):=\inf_{\|g-f\|\leq r}\,\max_{x\in g^{-1}(0)} o(x)$ as the $o$-filtration value of the lowest nonzero entry of the right-hand side after the reduction by the column of filtration value $r$. These values are just a side product of the matrix reduction algorithm in Appendix~\ref{s:pers}.
The error in this approximation is bounded by the simplexwise Lischitz constants for $f$ and $o$.

\section{Experimental results}\label{s:experiment}
\label{s:experiments}
\heading{Motivation.}
One motivation for implementing the algorithm was to experimentally analyse the following question: 
\begin{itemize}
\item How typical is a~situation in which the zero cannot be detected by primary obstruction and higher obstructions are needed?
\end{itemize}

To illustrate the flavour of this problem, consider a~function $f$ from an $(n+1)$-ball $B^{n+1}$ to $\R^n$ such that $0$ is a~regular value of $f$ and
the zero set is a~circle. If $r$ is small, then the $r$-neighborhood of the zero set is homeomorphic to a~solid torus $S^1\times B^n$.
An $n$-hyperplane intersecting the zero set transversally will typically 
intersect this torus in a $n$-disc $\{*\}\times B^n$ with a~zero of $f$ inside: this reflects the non-extendability to the $n$-skeleton.
However, with increasing $r$ (and hence increasing our freedom to perturb the function), 
the primary obstruction will die once the $r$-neighborhood touches the boundary or becomes a full $(n+1)$-ball:
in the latter case, a~nontrivial secondary obstruction is reflected by the homotopy class of the map from the boundary of this 
$(n+1)$-ball to $S^{n-1}$. This homotopy class is encoded in the gradient-induced framing of the original zero set of $f$: if the framing 
is trivial (framed null-cobordant), then higher obstructions don't occur. If the framing is ``twisted'', then they do.

Intuitively, we assumed that using Gaussian random fields, the gradient-induced framing of the zero set should be quite random and we would observe
twisted as well as untwisted cases. Experiments, however, do not support this so far, which we find surprising.

\heading{Description of the computation experiments.}
The lowest-dimensional case where nontrivial secondary obstruction can occur is $\dim X=4$ and $n=3$. 
Using an experimental approach, we generated random continuous functions from a~regular $4$-dimensional cubical grid into $\R^3$ 
taken from different probability distributions. 
The space $X$ was either a $4$-cube or a $4$-torus $(S^1)^4$ and
the underlying simplicial complex was the Freudenthal triangulation of the canonical cubical subdivision of $X$~\cite[p. 154]{allgower2003introduction}.
Instead of $A_r$ from Definition~\ref{d:A-r}, we used a coarser filtration $A_r^\square$ based on the cubical structure,
see Appendix~\ref{s:implementation} for details.
We computed the vertex-approximation $f'$ from Definition~\ref{d:A-r} and the smallest $r_0>0$ such that
$f'$ is simplicial on $A_{r_0}^\square$. Then we found the persistence of the primary obstruction $r_1\geq r_0$ and the persistence
of the secondary obstruction $r_2\geq r_1$: the goal was to check whether instances with $r_2>r_1$ occur and how often. 

First we experimented with Gaussian random fields. 
Such functions are continuous and infinitely differentiable~\cite[Sec. 2.2]{Adler:1981}. 
For each component $f_i$ of $f$ and each vertex $x$, 
the random variable $f_i(x)$ was normalized to the standard normal distribution $N(0,1)$ 
and the covariance between $f_i(x)$ and $f_i(y)$ was taken to be $C(x,y)=\tilde{C}(|x-y|)$: we tried different functions $C$.
First we generated random functions such that the 
discrete Fourier transform of $C(0,x)$ was proportional to $((1+|p|^2)^{-l})_{p\in \{0,\ldots, g-1\}^4}$ for various constants $l$ 
(compare~\cite[p. 12]{Lang:2011}). The value $l=0$ corresponds to white noise and $l=\infty$ to constant functions.
While this procedure naturally creates functions on a torus, for experiments on a cube  we generated a~random function
on the discrete torus $\{1,2,\ldots, 2g\}^4$ and restricted it to the coordinates $\{1,\ldots, g\}^4$ to avoid periodicity.
The three components of $f$ were generated independently. To assure that the resulting function has zero at all, we analyzed the function
$f(x)-f(x_0)$ instead of $f(x)$, where $x_0$ was the midpoint of the cube, resp. a~fixed point in the torus.

In most cases, we detected a nontrivial primary obstruction, but not a single instance with secondary obstruction $r_2>r_1$.
To give an illustration, the following table shows some statistics of one of the experiments on a $4$-cube:
$l$ is the parameter of the distribution,
$g$ is the number of vertices in each dimension, 
$r_0$ the smallest value for which $f'|_{A_{r_0}^\square}$ is simplicial,
$r_1$ the average persistence of the primary obstruction in cases when $r_1>r_0$, 
and max. $r_1$ the largest persistence of primary obstruction. 
The averages are taken out of 1000 functions for $l\in \{3, 3.5,4, 4.5\}$ and out of 10\,\,000 for $l=5.0$.
\\

\begin{tabular}{c|c|c|c|c|c}
$l$ & $g$ & $r_0$ & \% of $r_1>r_0$ & average $r_1$ if nontrivial & max. $r_1$\\
\hline 
3.0  & 30 & $0.4$ & 78\% & $0.56$ & $0.95$ \\
3.5  & 30 & $0.21$  & 91\% & $0.41$ & $0.82$\\
4.0  & 25 & $0.15$ & 91\% & $0.32$ & 0.66 \\
4.5  & 25 & $0.1$ & 91\% & $0.25$ & $0.55$ \\
5.0  & 20 & $0.1$ & 87\% & $0.21$ & $0.63$ \\
\end{tabular} \\

When performing such experiments on the $4$-torus, it sometimes happened that the cup square of a computed extension $x\in\Omega(A_{r_1})$
was nontrivial in $H^4(X,A_{r_1}^\square)$, giving some ``hope'' of a nontrivial secondary obstruction: however, in all cases, this could be removed 
after replacing $x$ by another extension of the pullback $y=(f'|_A)^*(z)$ to the $2$-skeleton 
(see Section~\ref{s:obstructions}).\footnote{In fact, nontriviality
of the secondary obstruction on a $4$-torus can only be reduced to a~system of quadratic Diophantine equations. While we cannot 
algorithmically check satisfiability of quadratic equations, 
in all cases where we had to deal with this problem, these equations were almost trivial and solvable.}

In other rounds of experiments, we generated functions from a $5$-torus into $\R^4$ or
replaced the correlation function $C(x,y)$ by the Gaussian function
$$\exp\left(-\frac{|x-y|^2}{2l^2}\right)$$ for suitable $l>0$, but the results were were similar to that from the distribution above.

In another attempt to detect secondary obstruction in random fields we generated random homogenous quadratic polynomials on $[-1,1]^4$.
The coefficients $a_{ij}^k$ in $f_k(x)=\sum_{i,j} a_{i,j}^k x_i x_j$ were independent samples from a standard normal 
distribution.\footnote{This is motivated by the fact that the simplest examples of functions with nontrivial secondary obstruction are quadratic
and homogenous.}
The zero set of homogenous quadratic functions is either the origin alone or a cone intersecting the boundary $\partial [-1,1]^4$:
only the first case can yield a nontrivial $H^4(X,A_r^\square)$ and a nontrivial secondary obstruction. 
We generated around 70 thousand instances of random quadratic functions on a $10^4$ grid: around 2.2\% of them
had only the origin as the zero set, but there was no nontrivial secondary obstruction in a single instance. 

\heading{Possible explanations.}
One observation related to the lack of secondary obstruction is that the cohomology in dimension $n+1$ has typically 
lower persistence than in dimension $n$ and most generators have already died when the primary obstruction (element of~$H^n$)
dies. A similar phenomenon occurs in persistent homology of excursion sets of random scalar fields, where the persistence
barcodes in dimension $0$ die before the barcodes in dimension 1, compare~\cite{Adler:2010}. 
In the vast majority of our experiments on the $4$-cube, the $4$-dimensional cohomology group $H^4(X,A_{r_1}^\square)$ 
was trivial for $r_1$ being the persistence of primary obstruction.
The lack of top dimensional cohomology in this case probably reflects the fact that most components of the neighborhood of the zero set intersect 
the boundary of the domain, although this argument does not apply for the torus. 

Another remark possibly explaining the lack of secondary obstruction is the following idea. If the codimension is one, such as in our experiments, then the 
generic zero set is a~union of circles. The presence of a non-trivial secondary obstruction implies that the gradient-induced framing on the the zero set is not
framed null-cobordant in $X\setminus A$ (see~\cite[Thm C]{Persistence_of_zero}). For any circle in the zero set of $f$, either the circle is small, or it is large. Derivatives of random Gaussian fields are
themselves random Gaussian fields and hence, if the circle is small, then the framing vectors are more likely to be close to constants and hence ``untwisted''. In the other extreme,
if the circle is large, then the framing may  be twisted, but it is quite likely that any filler of the circle contains ``large'' values of~$f$. But then the primary obstruction $r_1$ may
be  large enough to ``outvoice'' a~potential secondary obstruction $r_2>r_1$: namely,  $X\setminus A_{r_1}$ may become so large that the framing is already null-cobordant there. 

Laying  down the groundwork for a~solid theory which would explain this phenomenon is the subject of future research.

\heading{Experiments with formulas.}
Another motivation for implementing the algorithm was to test the running time and memory limitations in practice. 
Our testing benchmarks consisted of inputs in which the function values $f(v)$ were generated via formulas with known properties 
in a cubical grid. We ran many testing examples, some of them being shown in Appendix~\ref{s:formulas}.
To summarize the results, the performance is much better than the worst-case complexity bound derived in Section~\ref{s:complexity}
and is approximately linear in the number of simplices of the input. We were able to run benchmarks up to $\dim X=8$ for small grids,
such as $5^8$: the largest coboundary matrix for which we computed a nontrivial obstruction had 40 million columns.

In higher dimension, the main obstacle is  the size of the input rather then the complexity of our algorithm. 
It is an interesting open question
whether some different approach exists for approximating the robustness of zero  in high-dimensional spaces, provided that the input has a~``small'' format, such as an explicit system of
equations given by formulas.

\section{Complexity}
\label{s:complexity}
The input size (and hence computational complexity) depends heavily on the encoding of the simplicial complex.
For example, we may 
specify the set of \emph{all} simplices, or the set of all top-dimensional simplices.\footnote{In other situations,
the input specifying the simplicial complex could be even smaller. 
One example is specifying the vertex set in $\R^m$ and assuming the Delaunay triangulation.}
Therefore we study \emph{parameterized complexity} as a function of the following parameters.
Let $m$ be the dimension of $X$ and $n\leq m$ the dimension of the target space $\R^n$.
We define $N$ to be the maximum of the number of $k$-simplices
for $k\in \{n-1,n,n+1\}$ and $V$ the number of vertices. 
In addition to specifying $X$, the input contains the function values in all vertices, that is, $n\times V$ numbers.

We present complexity bounds as a~function of $N,V,m$ and $n$. 

\heading{Primary obstruction.}
We assume that the function values $f(v)$ at the vertices are all rational vectors and that we can compare their absolute values 
$|f(v)|$, $|f(w)|$ in unit time (these numbers may be roots of rational numbers for $\ell_p$ norms).
Then computing the vertex approximation $f'(v)$ for each vertex $v$ via~(\ref{d:vertex-app}) amounts to $O(nV)$ operations.
Computing the filtration of all $(n-1)$-simplices via formula (\ref{e:filtration}), as well as the pullback $y$ and its codifferential
 are by definition subroutines of complexity $O(nN)$; ordering the $(n-1)$-simplices by filtration is done in $O(N\log N)$.
The computation of the codifferential matrix is again of order $O(nN)$ if we store it in a~sparse format, 
because each row of the matrix corresponds to the boundary of an $n$-simplex and has only $n+1$ nonzero elements. 

The bottleneck of computing the primary obstruction is the EARLIEST SOLUTION algorithm described on page~\pageref{p:es}. 
An implementation based on a~binary or exponential search requires at most $\log N$ 
solutions of a linear system of Diophantine equations. Each of them is a system of at most $N$ rows and columns, all coefficients being $\pm 1$ or $0$.
By~\cite[Thm. 19]{storjohann1996fast+} we may solve any such Diophantine system in $O(N^4\,\log^4 N)$ time, which yields 
$O(N^4\,\log^5 N+n(N+V))$ as a complexity bound for the primary obstruction. 
Assuming the lack of blowup of matrix coefficients during the matrix reduction, we can bound the number of arithmetic operations in EARLIEST SOLUTION  by $O(N^3)$. This is discussed in more detail in Appendix~\ref{s:pers}. 
In this scenario, sub-cubic bounds could be achieved using randomization~\cite[Thm. 39]{storjohann2005shifted}.
In practice, however, our implementation
of~EARLIEST SOLUTION exhibits subquadratic scaling, allowing us to experiment with instances
for $N \le 10^7$. This is not entirely surprising---large instances of simplicial boundary matrices
are commonly reduced in the field of computational topology. 

\heading{Secondary obstruction.}
The bottleneck in the secondary obstruction algorithm is the computation of all Steenrod squares of all the generators of $H^{n-1}(X,A_r;\Z)$ 
for all filtration values $r$. In a naive implementation we may compute a set of generators of $Z^{n-1}(X;\Z)$ 
and their respective filtration values.
Generators of the kernel (over $\Z$) of a~matrix with at most $N$ rows and columns can be computed in 
$O(N^4)$~\cite[Theorem 1]{Buchmann:99}.
The number of such generators is bounded by $N$. In the Steenrod square computation,
we need to compute, in the worst case, the values on all $(n+1)$-simplices; in each evaluation, the formula for $\smile_{n-3}$
described in~\cite{Steenrod} contains an iteration over all elements of $m\choose 4$ (Steenrod pairs).
Thus, computing the Steenrod squares of the generators of $Z^{n-1}(X;\Z)$ is $O(N^2\,m^4)$.
The final matrix computation corresponding to equation (\ref{e:second-pers3}) is done over the field $\Z_2$ 
which only requires $O(N^\omega)$ operations for a constant $\omega<3$~\cite[Proposition 6]{jeannerod2013rank}.
This yields a~complexity bound of $O(N^4+N^2 m^4)$ for the persistence of secondary obstruction. For all practical purposes,
it is safe to assume that the values of $m$ can be ignored.

\subsection*{Acknowledgements}
We thank Robert Adler for the discussion on random Gaussian fields, and Eric Wofsey for his hints on math.stackexchange 
regarding the triviality of the cup products $H^2(X,A)\times H^2(X,A)\to H^4(X,A)$ for contractible $X$~\cite{Wofsey-stack}.
Further, we thank both Institute of Computer Science of the Czech Academy of Sciences as well as IST Austria for providing
 computer power for our computational experiments.

\bibliographystyle{spmpsci}      

\bibliography{../Postnikov,../sratscha}

\newpage
\appendix

\section{Secondary obstruction.}
\label{s:secondary-persistence}
\heading{Persistence of the secondary obstruction---the algorithm for $n>3$.} 
Assume that a filtration $X\supseteq A_{r_0}\supseteq A_{r_1}\supseteq \ldots$ and a simplicial map $f': A_{r_0}\to\Sigma^{n-1}$
are given, $n>3$, and vertices on $X$ and $\Sigma^{n-1}$ are ordered so that $f'$ is order-preserving: this order is used
in the implementation of the $\smile_{n-3}$ operation on the level of cochains.
Further, we assume that the persistence of primary obstruction $r_j$ has already been computed by the algorithm
described on page~\pageref{page:primary-persistence}.
That is, the restriction of $f$ to $A_{r_j}$ is not extendable to some continuous map $A_{r_j}\cup X^{(n)}\to\Sigma^{n-1}$, but the restriction to $A_{r_{j+1}}$ is extendable. 
We continue to use the notation of Section~\ref{s:obstructions}: in particular, $z$ is the characteristic cocycle of a~fixed
$(n-1)$-simplex in $\Sigma^{n-1}$, $\bar y\in C^{n-1}(X,A_{r_0})$ is a cochain extending the pullback $y=(f')^\sharp(z)\in Z^{n-1}(A_{r_0})$ 
of $z$ and  $\emptyset\neq \Omega(A_{r_{j+1}})$ is the set of all $(n-1)$-cocycles on $X$ that extend $y$ on $A_{r_{j+1}}$. 

By Proposition~\ref{p:secondary}, the persistence of secondary obstruction is the largest number $r_k$ such that
\begin{equation}
\label{e:second-pers1}
\delta c=v(x), \quad  c\in C^n(X,A_{r_k};\Z_2),\,\, x\in \Omega(A_{r_k})
\end{equation} 
has no solution (where $v$ is defined by (\ref{e:steenrod})).

Let $x\in\Omega(A_{r_{j+1}})$ be a~fixed extension of $\bar y$, computed in the algorithm for primary persistence. 
Then also $x\in\Omega(A_{r_k})$ for each $k > j$.
For any such $k$, $\Omega(A_{r_k})$ is a coset in $Z^{n-1}(X)$
$$
\Omega(A_{r_k})=x+Z^{n-1}(X,A_{r_k};\Z)
$$
and hence equation (\ref{e:second-pers1}) reduces to
\begin{equation}
\label{e:second-pers2}
\delta c=v(x-w), \quad \text{for some } w\in Z^{n-1}(X,A_{r_k};\Z), \,\, c\in C^n(X,A_{r_k};\Z_2).
\end{equation} 
The crucial property we will use is that $v(x-w)$ is a relative coboundary iff $v(x)-v(w)$ is a relative coboundary:
this follows directly from the linearity of the Steenrod square operation $H^{n-1}(X,A;\Z_2)\to H^{n+1}(X,A;\Z_2)$ for $n>3$.
Thus we can reformulate (\ref{e:second-pers2}) to the problem of finding the maximal $r_k$ such that
\begin{equation}
\label{e:second-pers3}
\delta c+v(w) = v(x), \quad  w\in Z^{n-1}(X,A_{r_k};\Z), \,\, c\in C^n(X,A_{r_k};\Z_2).
\end{equation}
has no solution.
To simplify the computations, we don't need to consider all cocycles $w\in Z^{n-1}(X,A_r)$ 
but only generators of the cohomology $H^{n-1}(X,A_r;\Z)$: the Steenrod square of any relative coboundary is again a relative 
coboundary $\delta c'$ for some $c'\in C^{n}(X,A_r;\Z_2)$, so adding it has no impact on the solvability of (\ref{e:second-pers3}).

The right-hand side of (\ref{e:second-pers3}), $v(x)$, is a cocycle that does not 
depend on $k$ (assuming $k>j$). The left-hand side is a~combination of coboundaries of characteristic cocycles of $n$-simplices $\Delta^n$
and cochains of the form $v(w)$ for $(n-1)$-cocycles $w$. To each $\Delta$ and $w$ is assigned a filtration value and 
we want to minimize the value $r_k$ such that $v(x)$ \emph{can} be expressed as a~combination of $\delta c$'s and $v(w)$'s
such that $c$ and $w$ have  filtration values $\leq r_k$, but \emph{cannot} be expressed as a~combination of such cochains with filtrations of
$c$ and $w$ \emph{strictly smaller} than $r_k$.

Summarizing the above steps, we obtain the following algorithm:
\begin{itemize}
\item Order the vertices  of $\Sigma^{n-1}$ and the vertices of $X$ so that $f$ is order-preserving
\item For a precomputed $x\in\Omega(A_{r_{j+1}})$, compute the relative cocycle $v(x)\in Z^{n+1}(X,A_{r_{j+1}};\Z_2)$ by the 
definition of the Steenrod operation $\smile_{n-3}$
\item Compute a subset $W\subseteq Z^{n-1}(X)$ that contains all cohomology generators $w$ of all $H^{n-1}(X, A_r;\Z)$ for all $r > r_j$.
(It may be a~set of generators of $Z^{n-1}(X)$.) To each $w\in W$, assign its filtration value $r(w)$ 
to be the minimal $r$ such that $w$ is zero on $A_r$.
\item Compute the filtration value of all $n$-simplices, using (\ref{e:filtration}). 
\item Order all $n$-simplices and elements of $W$ by their filtration.
\item Choose a basis of $(n+1)$-simplices and express the right-hand side $v(x)$ as a~vector $\aa\in (\Z_2)^q$. 
\item Create a matrix $M$ whose column set consists of
\begin{itemize}
\item all coboundaries of characteristic cochains of $n$-simplices
expressed in the basis of $(n+1)$-simplices over $\Z_2$
\item all elements $v(w)$ expressed in the basis of $(n+1)$-simplices
\end{itemize}
\item Order the columns of $M$ by filtration.
\end{itemize}
Computing the persistence of the secondary obstruction then reduces to solving the EARLIEST SOLUTION problem for $M\xx = \aa$, 
this time over the $\Z_2$-field.

The hardest part is to compute the cohomology generators: this algorithm is summarized as follows.

\framedpar{
{\bf Problem {\sc Persistent Generators:}}
\\
Input: A filtration $X\supseteq A_0\supseteq\ldots\supseteq A_h=\emptyset$.
\\
Output: A sequence $g_1,\ldots,g_\nu\in Z^{n-1}(X;\Z)$ and a sequence of integers $\mu(0)\le\ldots\le\mu(h)$ such that $[g_1],\ldots, [g_{\mu(i)}]$ generate $H^{n-1}(X,A_i;\Z)$ for each $i$.
}
We give imeplementation details of this part on a lower-level in Section~\ref{s:persistent-gen}.
Our algorithm will give the output with the number of generators $\nu$ minimal in a certain sense. Notably, when working over $\Q$ or $\Z_p$ instead of integers, the output would correspond to a persistence bar code with all the death information erased but including representative (co)cycles for each bar.
 
\heading{The special case $f: [0,1]^4\to\R^3$.}
\label{h:mn=43}
This section justifies the claim that for $X=[0,1]^4$ and $n=3$,
the general algorithm for computing persistence of secondary obstruction works, if we completely ignore the persistence generators $w$
and perform all computations over $\Z$-coefficients.

Assume first that $n=3$ and $X$ is arbitrary. The two main differences compared to the algorithm above are:
\begin{itemize}
\item The Steenrod operation and final coboundary matrix have to be computed over $\Z$, not $\Z_2$. This goes back to the fact that
the homotopy group $\pi_3(S^2)\simeq \Z$, unlike $\pi_n(S^{n-1})\simeq \Z_2$ for $n>3$. The homotopy group serves as cohomology 
coefficients in the theory of obstructions.
\item The operation $\smile_{n-3}$ reduces to the cup product $\smile$ and the operation $w\mapsto w\smile w$ is not 
linear on the level of cohomology but quadratic.
\end{itemize}
This means that while for any particular extension $x\in \Omega(A)$ of $\bar y$, we may test satisfiability of $\delta c=v(x)$,
we cannot test the \emph{existence} of such $x$ using a linear system of equations.
However, if $X=[0,1]^4$, we claim that for any two extensions $x,y\in\Omega(A)$ of $\bar y$, $v(x)-v(y)$ is a relative coboundary
and thus we need to check the equation $\delta c=v(x)$ only for one $x$.
\begin{lemma}
\label{l:trivial-cup}
Let $X$ be a triangulation of $[0,1]^4$, $A\subseteq X$, $x\in Z^{2}(X)$ and $w\in Z^2(X,A)$. Then 
$$
(x-w)\smile (x-w) \,-\, (x\smile x) \in B^4(X,A).
$$
\end{lemma}
Using the parametrization $\Omega(A)=\{x-w:\,\,w\in Z^{n-1}(X,A;\Z)\}$ we immediately obtain that our general algorithm
for the persistence of secondary obstruction works once we replace the $\Z_2$-coefficients by $\Z$-coefficients in its final step. 
Moreover, we may completely ignore the persistent generators $w$ and don't need to compute them at all.\footnote{Note that
$x\smile x$ may represent a nontrivial element of $H^4(X,A)$, as $x$ is not an element of $Z^2(X,A)$ in general. The fact that
$x\smile x$ is zero on $A$ follows from the fact that $f$ is order-preserving.}
\begin{proof}
By bilinearity of the cup product, $(x-w)\smile (x-w) \, - \, x\smile x\,= \, - \, x\smile w \, -\, w\smile x\, +\, w\smile w$.
The mixed-term $x\smile w$ is a relative coboundary, because $\smile$ induces a bilinear product on the level of cohomology
$H^2(X)\times H^2(X,A)\to H^4(X,A)$ and $H^2(X)=0$, as $X$ is contractible. It remains to show that $w\smile w$ is a~relative coboundary.
Let $CA$ be the cone over $A$ and $A\hookrightarrow CA$ be an inclusion.
The inclusion $A\hookrightarrow X$ can be extended to a map $CA\to X$, because $X$ is contractible. The map of pairs
$(CA,A)\to (X,A)$ induces the commutative diagram in which the rows are the long exact sequences of cohomology groups.
$$
\begin{array}{ccccccccc}
H^{*-1}(X) & \to & H^{*-1}(A) & \to & H^*(X,A) &\to & H^*(X) &\to & H^*(A) \\
\downarrow & &\downarrow & &\downarrow& &\downarrow& & \downarrow\\
H^{*-1}(CA) & \to & H^{*-1}(A) & \to & H^*(CA,A) &\to & H^*(CA) &\to & H^*(A) 
\end{array}
$$
The vertical arrows $H^{*(-1)}(X)\to H^{*(-1)}(CA)$ are trivial as both spaces are contractible and 
$H^{*(-1)}(A)\to H^{*(-1)}(A)$ are identities.
By the five-lemma~\cite[p. 129]{Hatcher}, the middle homomorphism $H^*(X,A)\to H^*(CA,A)$ is an isomorphism. 
Further, $H^j(CA,A)\simeq H^j(CA/A)$ for $j>0$. The space $CA/A=:\Sigma A$ is the \emph{suspension} of $A$.
The cup product $H^2(\Sigma A)\times H^2(\Sigma A)\to H^{4}(\Sigma A)$ 
of a~suspension is trivial~\cite[Corollary 4.11]{Bredon-at} and the naturality of cup product implies that the cup product 
$H^2(X,A)\times H^2(X,A)\to H^{4}(X,A)$ is trivial for $j>0$ as well. This shows that $w\smile w\in B^4(X,A)$. 
\end{proof}

\section{Persistent integral homology computations}\label{s:pers}
\subsection{Algorithm for the \textsc{Earliest Solution} problem}
\label{s:earliest}

The {\tt earliest\_solution} algorithm is used to find the persistence of a (co)cycle. 
This is closely related to computing persistent homology, which is a well-studied problem, at least for coefficient in a finite field. 
We adapt the boundary matrix reduction algorithm by 
Edelsbrunner, Letscher and Zomorodian~\cite{edelsbrunner2002}, originally developed for persistent homology. Note that this algorithm, unlike classical 
Gaussian elimination or Smith normal form algorithms, is incremental, which is required by our application. 
Moreover, it only uses column operations, making efficient implementation for sparse matrices relatively easy.

\heading{Efficiency over finite fields.}
Recent work on computing persistent homology over finite fields resulted in significant performance improvements, see~\cite{phat,gudhi}. Despite the cubic worst-case bound, linear scaling is achieved on practical datasets, involving sparse matrices of size $10^9 \times 10^9$ and more. This encouraged us to adapt the modern version of the 
classical persistence algorithm to solve our problem over the integers, rather than adapting 
classical Smith or Hermite normal form algorithms to the persistent setting.

\heading{Reduced form and reduction.}
We adapt the notation common in computational topology literature: the \emph{lowest nonzero} of a nonzero column is defined as the lowest position (largest index) with nonzero coefficient. A sub-matrix is called \emph{reduced} if all the lowest nonzeros are unique. In other words, given a lowest nonzero, there may be other nonzero entries in the same row, but they must not be \emph{lowest} nonzeros. When this invariant is not satisfied, we say there is a collision. By \emph{lowest value}, we refer to the value of the lowest nonzero entry.

The algorithm starts from an empty matrix and adds one column at a time, maintaining the reduced prefix of the matrix, $R$. The rightmost column of each prefix is called the current column. 

Procedure {\tt reduce\_column} reduces the column {\tt curr} with respect to the reduced prefix {\tt R}.

\begin{lstlisting}
def reduce_column(current, R, force_divisibility=False):        
        while curr collides with some column in R:
                coll = column in R colliding with curr
                P = lowest_value(coll)
                Q = lowest_value(curr)             
                if force_divisibility is True and P does not divide Q:
                        return curr
                use Euclid to find nonzero a,b,c,d s.t. 
                    gcd(a,b) == 1 and gcd(c,d) == 1 and
                    a*P + b*Q == gcd(P,Q) and
                    c*P + d*Q == 0                          
                coll, curr = a * coll + b * curr, c * coll + d * curr                
        return curr
\end{lstlisting} 

Procedure {\tt earliest\_solution} solves the stated problem.
\begin{lstlisting} 
def earliest_solution(M, a):
        augment M with identity matrix (above)
        augment a with zeros (above)
        R = []
        for col in M:
                reduced = reduce_column(col, R, force_divisibility=False)
                append reduced to R            
                reduce_column(a, R, force_divisibility=True)            
                if a is zero:
                        return change of basis of a 
        return None
\end{lstlisting} 

\heading{Correctness.}
The basic operation is the addition of two different columns, without affecting the column span of the relevant matrix prefix. Therefore the solution is unaffected.

To retrieve the solution, at each step we attempt to reduce the input column vector {\tt a} with respect to the currently reduced prefix {\tt R}. The solution exists iff {\tt a} becomes zero, and is encoded by (the negation of) the change of basis column of {\tt a}. 
Since we solve the equation $Ax = a$ (and not $Ax = ka$) we perform additional divisibility check:
see {\tt force\_divisibility} in the {\tt reduce\_column} procedure.

\heading{From finite fields to integers.}
In the integral case, we may modify both the current column and the colliding column. This is in contrast with the finite field case, in which 
only the current column is modified. To determine the required linear combination of columns, we use the extended Euclid algorithm. 
As a result, the lowest \emph{value} of a certain column might change (decrease) many times during the reduction of other columns, but the \emph{position} is fixed once the column is reduced. Because of this, while reducing column {\tt a}, we need to take into account 
previously reduced columns, and not only the current one. Moreover, after a colliding column
is affected, it may not be in the column span of the prefix of the original matrix ending at this column. At this stage, however, this column may only affect columns that succeed the currently reduced column. Therefore, correctnes of the algorithm is unaffected.

\heading{Efficiency.} For efficiency reasons we use one technique suggested in~\cite{phat}: The current column is stored in a data-structure handling fast column additions and maximum element queries. One natural choice is a balanced binary search tree; 
more efficient alternatives are available. This way we avoid the following common bad case: Let $n$ be the total number of columns in $M$. When an current column becomes dense (the number nonzero entries is $\Theta(n)$), adding $\Theta(n)$ sparse columns takes time quadratic in $n$. Avoiding this situation does \emph{not} imply that we can efficiently handle matrices that become dense due to fill-in. However in practice, often a small number of columns display this behavior.

We also perform the computations in an on-line fashion -- we read columns one by one and stop once a solution is found. This gives significant practical improvements, because often the necessary matrix prefix is very small compared to the matrix of the entire complex.

Overall the algorithm performs well, exhibiting roughly linear scaling in the number of nonzero entries of the original matrix. In particular, we didn't observe coefficient blowup. 
Note that the lowest nonzero position in the current column decreases after resolving each collision, so the number of collision resolutions is quadratic in $n$. Each such resolution requires combining two columns, potentially of size $O(n)$, due to fill-in. This leads to cubic worst case running time bound, assuming that the magnitude of coefficients can be bounded by a constant. Of course, there exist cases when the blowup does occur, and the above analysis does not apply. More advanced algorithms can be used to alleviate the effect of the blowup, possibly at the cost of simplicity and efficiency in the situations that we encountered thus far. 

\subsection{Algorithm for the \textsc{Persistent Generators}
problem}
\label{s:persistent-gen}
After possibly refining the sequence $(A_i)_i$ we may assume that for
each filtration value $i$ there is exactly one $(n-1)$-simplex
$\Delta_i^{n-1}$ of that filtration value, that is $\Delta_i^{n-1}
\in A_{i-1}\setminus A_{i}$ (and possibly several $(n-2)$-simplices
in $A_{i-1}\setminus A_{i}$). This makes the algorithm and its analysis simpler.

The algorithm for \textsc{Persistent Generators} problem follows.
By the statement ``reduce the column'' we refer to the procedure
\texttt{reduce\_column} from Section~\ref{s:earliest}
\begin{itemize}
\item Let $M$ be the coboundary matrix for $\delta^{n-1}\:C^{n-1}(X;\Z)\to
C^n(X;\Z)$, that is, it consists of columns $\mm$ for integral
combination for $\delta \Delta^{n-1}$ for each $\Delta^{n-1}\in
X$ of growing filtration values. In exactly the same way we create
the coboundary matrix $N$ for $\delta^{n-2}$. During the reduction
of the matrix $M$ we keep track of the change of basis. (On low
level, we augment the matrix $M$ and then the change of basis
vector for a given column is encoded in the augmented part of
that column.)

\item We initialize $G:=()$ to be an empty sequence of $(n-1)$-cocycles
on $X$ and $\mu(-1):=0$.
\item For each filtration value $i=0,\ldots,h$ do the following:
\begin{itemize}
\item Reduce all columns of $N$ of the filtration value $i$.
\item Reduce the column $\mm$ of $M$ of filtration value $i$
(i.e., corresponding to the cochain $\delta \Delta^{n-1}_i$).
Let $\gg$ be the change of basis vector after the reduction.
\item If the reduced column $\mm$ equals to $0$ and there is
no reduced column $\nn$ in $N$ such that $\lowest(\nn)=i$ and
$\nn_{\lowest(\nn)}|\gg_{\lowest(\gg)}$ then add $\gg$ to $G$
and set $\mu(i):=\mu(i-1)+1$.
\item Otherwise set $\mu(i):=\mu(i-1)$
\end{itemize}
\item Output $G$ and $\mu(1),\ldots,\mu(h)$.
\end{itemize}

\begin{theorem}
The above algorithm solves the \textsc{Persistent Cycles} problem. Namely, after its $i$th iteration, the cohomology group $H^{n-1}(X,A_i)$ is generated by $[g_1],\ldots,[g_{\mu(i)}]$ where cycles $g_1,\ldots,g_{\mu(i)}\in Z^{n-1}(X,A_i;\Z)$ correspond to the columns
$\gg$ added to $G$ during the iterations $0,\ldots,i$.
\end{theorem}
\begin{proof}
We proceed by induction on the value $i$. For $i=-1$ the claim trivially holds.

Let us assume that $i\ge 0$ and that the claim holds for $i-1$. First observe that once the new column $\mm$ is not reduced to a zero column, then $Z^{n-1}(X,A_i;\Z)=Z^{n-1}(X,A_{i-1};\Z)$. Otherwise $Z^{n-1}(X,A_i;\Z)=Z^{n-1}(X,A_{i-1};\Z)\oplus \langle g\rangle$ where the cocycle $g\in Z^{n-1}(X,A_i;\Z)$ corresponds to the change of basis vector $\gg$ in the $i$th iteration. We only have to check whether $[g]$ is a linear combination of  $[g_1], \ldots, [g_{\mu(i-1)}]$. By the induction hypothesis, it happens if and only if $g$ is cohomologous to a cocycle in $Z^{n-1}(X,A_{i-1};\Z)$. And this in turn is true if and only if we can reduce to zero the lowest nonzero component (that is, the $i$th component) of the vector $\gg$ by adding a combination of columns from $N$ of filtration value at most $i$ (these columns generate the group of coboundaries $B^{n-1}(X,A_{i};\Z)$). Since this is exactly the reduced part of the matrix $N$, it is enough to find if there is a column with the lowest nonzero index equal to $i$ and check the divisibility condition.  
\end{proof}

\section{Some details of our implementation: exploiting the cubical structure}
\label{s:implementation}
The domain $X$ we consider in the implementation is a cube $[0,1]^m$ triangulated as follows. 
We define $I_1,\ldots, I_m$ to be unit intervals subdivided into $n_i$ equidistant intervals of length $1/n_i$, $i=1,\ldots, m$.
This yields a cubical set structure $\prod_j I_j$ on the unit cube. Further, we subdivide each $m$-cube into $m!$ simplices
via the Freudenthal triangulation~\cite[p. 154]{allgower2003introduction}.  The resulting triangulation
naturally corresponds to the product $\prod I_j$, understood as a product $\prod_j I_j$ of \emph{simplicial sets}~\cite{Friedm08}.
We call the set of vertices a \emph{grid}. 
A function $f:X\to\R^n$ is then given by a set of $\R^n$-vectors in each vertex (a~multidimensional rasterized image), 
together with a~simplexwise Lipschitz constant $\alpha$.

Many operations in the algorithm outlined in the paper body---such as the computation of the pullback $y=f^\sharp(z)$, 
computing codifferentials of cochains and the $\smile_{n-3}$ operation---can be done \emph{locally}, 
without having to work with the full lists of simplices of a given dimension. 
The only step where global structure is needed is the {\textsc{Earliest Solution}} algorithm, 
where the matrix $M$ has rows indexed by the set of all $k$-simplices and columns indexed by $(k-1)$-simplices.
Indeed, the construction of this matrix and computation with it are the most memory- and time-consuming operations,
as the number of $k$-simplices increases rapidly with dimension.

Therefore, we implemented a modification of the algorithm described above based on the more feasible \emph{cubical} structure.
Instead of working with the simplicial filtration $A_r$, we use the filtration $A_r^\square\subseteq A_s^\square$ where 
$A_r^\square$ is defined to be the triangulation of the set of all cubes $c$ such that $|f(v)|\geq v$ for all vertices $v$ of $c$. 
These are still simplicial complexes and $f'$, $y$, the extension $\tilde{y}$ of $y$ and $\delta\tilde{y}$ are defined with no changes. For the computation
of persistence, however, we switch to the cubical setting via the Eilenberg-Zilber reduction~\cite{Eilenberg:1953}. 
We denote by $C^*(\Pi_j I_j)$ the simplicial cochains and by $\otimes_j C^*(I_j)$ the cubical cochains: there exist
chain homomorphisms of degree $0$
\begin{align*}
&AW: \otimes_j C^*(I_j)\to C^*(\Pi_j I_j)\\
&EML: C^*(\Pi_j I_j)\to \otimes_j C^*(I_j)
\end{align*}
such that $EML\circ AW$ is the identity and $AW\circ EML$ is chain homotopic to the identity. Both maps induce cohomology 
isomorphisms and can be relatively easily implemented using common formulas~\cite{Pedro:2005}. This allows us to switch between the simplicial
and cubical cochains anytime we need. Within the computation of persistence of the primary obstruction,
we compute the smallest $j$ so that there exists a cubical cochain $c_\square\in Z^{n-1}(X,A_{r_j}^\square)$ such that
$\delta c_\square=(\delta\tilde{y})_\square:=EML (\delta \tilde{y})$.
The matrix computation part \textsc{Earliest Solution} deals with the cubical coboundary
matrix, which is significantly smaller than the simplicial one.
For an illustration, the number of $k$ simplices in the triangulation of one $m$-cube is 50 already for $(m,k)=(4,2)$ and
more than 4 millions for $(m,k)=(8,6)$. 

For the secondary obstruction, we need to convert $c_\square$ back into a simplicial cochain and construct a simplicial $(n-1)$-cochain 
$x\in \Omega(A_{r_j})$ that extends $y$ on $A_{r_j}$ and is a global cocycle. 
We denote by SHI the cochain homotopy map $SHI: C^*(\Pi_j I_j)\to C^{*-1} (\Pi_j I_j)$, 
satisfying $AW\circ EML - \mathrm{id}=\delta\circ SHI + SHI\circ\delta$.
Then, using $\delta^2=0$, we obtain
\begin{align*}
\delta AW c_\square 
&= AW \delta c_\square = AW (\delta\tilde{y})_\square = AW\, EML\,\delta\tilde{y}= \\
& = \delta AW \,\,EML\, \tilde{y}=\delta \tilde{y} + \delta \,SHI \,\delta \tilde{y}.
\end{align*}
We compute $\tilde{c}:=AW (c_\square) - SHI(\delta \tilde{y})$ which is a \emph{simplicial} cochain zero on $A_{r_j}^\square$
and the last equation asserts that $\delta \tilde{c}=\delta\tilde{y}$ which allows us to compute the simplicial cocycle
$$
x:=\tilde{y} - \tilde{c} \in \Omega(r_j)
$$
having avoided to work with large lists of all simplices. 
The computation of $v(x):=(x\mod 2)\smile_{n-3} (x\mod 2)$ is done on the simplicial level and the property 
$x\in Z^{n+1}(X,A_{r_j}^\square ;\Z_2)$ depends on the fact that the chosen ordering of vertices of $X$ and $\Sigma^{n-1}$
is compatible with $f'$. Then we again apply the $EML$ operator to $v(x)$ and convert it to a cubical cochain $v(x)_\square$.
Persistent generators $w\in Z^{n-1}(X,A_{r}^\square;\Z)$ and their Steenrod images can be computed on the cubical level~\cite{Krcal2016}
and the cubical persistence of the secondary obstruction is done via a~cubical coboundary matrix, this time in dimension $n,(n+1)$.

There is a price to pay for the more convenient cubical filtration: it is courser than the simplicial one and
we can no more use the estimate on the robustness of zero set 
derived in Section~\ref{sec:algorithm-oracle}. On one hand, we have the relation $A_r^\square\subseteq A_r$.
This implies that whenever $f'|_{A_r}$ is homotopic to $f/|f|$, then so is $f'|_{A_r^\square}$, and
non-extendability of $f'|_{A_r^\square}$ to higher skeleta of $X$ implies the non-extendability of $f'|_{A_r}$.
Thus whenever the algorithm certifies non-extendability of $f|_{A_r^\square}$ and $r>\alpha n^{1/p}$, then 
$r-\alpha$ is a lower bound on the robustness of zero.
If we can additionally prove that $A_s\subseteq A_r^\square$, then the extendability of $f'|_{A_r^\square}$
implies the extendability of $f'|_{A_s}$.  The simplexwise $\alpha$-Lipschitz condition implies that any two vertices
$u,v$ of a cube $c$ satisfy $|f(u)-f(v)|\leq 2\alpha$, as $u,v$ can be connected by  path of at most two simplicial edges.
This implies $A_{r+2\alpha}\subseteq A_r^\square$ and extendability of $f'|_{A_r^{\square}}$ to all of $X$ implies the upper bound 
$r+3\alpha$ on the robustness.

Based on the experience with various functions in low dimensions, the practical performance has not one significant bottleneck.
The most resources-consuming steps in the primary obstruction computation include the computation of the cubical filtration and
the EARLIER SOLUTION subroutine, especially if the matrix has millions of columns.

\section{Experiments with formulas}
\label{s:formulas}
Our prototype is implementation in Python using numpy. While the efficiency is severely limited by this choice, 
we were able to run examples for $\dim X\leq 8$.

Although the strength of our methods is primarily in cases where we have uncertainty on $f$, 
the following examples provide some observations about the performance in cases where the
function values at the vertices are computed using exact formulas. 

In what follows, we assume a subdivision of the interval $[-1,1]$ into a set $I_g$ of $g$ equidistant points
and consider a grid $I_g^m\subseteq [-1,1]^m$ of size $g^m$. We chose the max-norm $\ell_\infty$ on $\R^n$ which yields
the smallest value of $\alpha n^{1/p}$ and allows smaller ``initial $r_0$''. This is desired, as for small grids it often happens
that $\alpha$ is large and $\alpha n^{1/p}$ may be larger then the robustness of zero in which case we fail to detect anything.

\heading{Testing primary obstruction on a quadratic function.}
First we used the function $f: [-1,1]^n\to\R^n$ given by 
\begin{align*}
& f_1(x)=x_1^2-x_2^2-\ldots -x_n^2 \\
& f_2(x)=2x_1 x_2\\
& \ldots \\
& f_n(x)=2x_1 x_n.
\end{align*}
This function has a single zero in the origin and has been used as a benchmark example in~\cite{nondec}. 
If $n$ is even, then the zero is robust and its robustness equals 
$\min_{x\in\partial [-1,1]^m} |f(x)|$: the zero has index two in this case. If $n$ is odd, then the zero has index zero 
and can be removed by arbitrary small perturbations. 
Using the max-norm, we use calculus to derive the estimate 
$|f(x)-f(y)|_{\infty}\leq 2n\,|x-y|_\infty$. Whenever $x,y$ are contained in one simplex of the triangulation, then the max-norm satisfies
$|x-y|_\infty\leq 2/(g-1)$ and we can use the piecewise Lipschitz constant $\alpha:=4n/(g-1)$. 

For our experiments, we now assume that $f$ is only given via a~list of function values in the grid $g^n$ and the 
simplexwise Lipschitz constant $\alpha=4n/(g-1)$. In this case, we have $\dim X=n$ so only primary obstruction can be nontrivial.
Using Lemma~\ref{t:approx}
for $p=\infty$ yields that whenever $r>\alpha$, then $f'$ is simplicial on $A_r$ and homotopic to $f/|f|$.
The relation $A_r^\square\subseteq A_r$ then implies that $f'$ is simplicial and homotopic to $f/|f|$ on $A_r^\square$ as well.
Thus we can run our algorithm with an initial $r_0=\alpha=4n/(g-1)$, compute the persistence of primary obstruction $r_1\geq r_0$
and---using the estimates at the end of Section~\ref{s:implementation}---conclude that whenever $r_1>r_0$ 
then the robustness of zero is between $r_1-\alpha$ and $r_1+3\alpha$.

The following table illustrates some properties of the computations: $n$ refers to the dimension, $g$ is the 
number of points subdividing the intervals in each dimension, the initial $r_0=\alpha$ is as above, 
$r_1$ is the persistence of the primary obstruction, 
$\delta^{n-1}$ the number of columns of the cubical coboundary matrix $C^{n-1}_\square\to C^n_\square$, 
the ``computed robustness'' displays lower and upper bounds $r_1-\alpha$, $r_1+3\alpha$ on the robustness of zero 
and the ``true robustness'' column the approximation of the real robustness of zero of the function $f$ 
defined exactly by the above formula,\footnote{This is $0$ in odd dimensions and $\frac{2}{\sqrt{n}+1}$ in even dimension.}
all wrt. the $|\cdot |_\infty$ norm.
\begin{center}
\begin{tabular}{c|c|c|c|c|c|c}
$n$ & $g$ & $r_0:=\alpha$ & $r_1$ & \# columns of $\delta^{n-1}$ & computed robustness & true robustness\\
\hline
$2$ & $20$ & $0.421$ & $0.8643$ & $760$ & $[0.443,2.127]$ & $0.828$\\
$2$ & $100$ & $0.08$ & $0.8285$ & $19 \,\,800$ & $[0.748,1.07]$ & $0.828$\\
$2$ & $500$ & $0.016$ & $0.83$ & $499\,\,000$ & $[0.814,0.878]$ & $0.828$\\
$3$ & $20$ & $0.632$ & $=r_0$ & $21\,\,660$ & $\leq 2.379$ & $0$\\
$3$ & $50$ & $0.245$ & $=r_0$ & $360\,\,150$ & $\leq 0.971$ & $0$\\
$3$ & $100$ & $0.121$ & $=r_0$ & $2\,\,940\,\,300$ & $\leq 0.483$ & $0$\\
$4$ & $20$ & $0.842$ & $=r_0$ & $548\,\,720$ & $\leq 3.37$ & $0.667$\\
$4$ & $30$ & $0.552$ & $0.711$ & $2\,\, 926\,\, 680$ & $[0.159, 2.367]$ & $0.667$\\
$4$ & $40$ & $0.41$  & $0.667$ & $9\,\,491\,\,040$  &  $[0.257, 1.897]$ & $0.667$
\end{tabular}
\end{center}
The total running times of these 9 computations is displayed in Figure~\ref{fig:times}.
\begin{figure}
\begin{center}
\includegraphics[width=12cm]{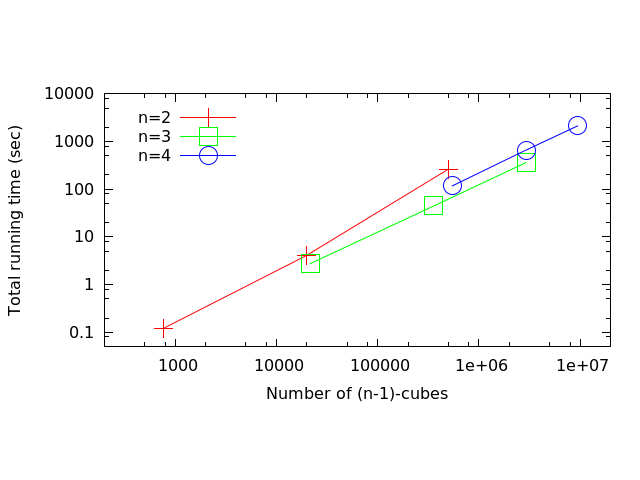}
\caption{Total running time of the computations, performed on Intel Xeon CPU X5680, 3.3 GHz,
as a function of the number of $(n-1)$-simplices (which equals the number of columns in the coboundary matrix).
}
\label{fig:times}
\end{center}
\end{figure}
The data suggest that the at least in these simple cases, the running time is approximately linear in the number of $(n-1)$-simplices.
The $n=3$ case is below the other two, because the primary obstruction is trivial there and the 
EARLIEST SOLUTION subroutine terminates almost immediately, using only very few columns in the matrix reduction.

In higher dimensions, the condition $r_0\geq \alpha:=4n(g-1)$ is too strict and we can only hardly start with $r_0$ that is smaller than the robustness.
However, we could take $r_0$ to be the minimal value for which $f'|_{A_{r_0}^\square}$ is simplicial:
in all cases that we have tried, we verified that $f$ was nowhere zero on $A_{r_0}^\square$ with $f'$ homotopic to $f/|f|$ 
on $A_{r_0}^\square$. 

Surprisingly, starting with minimal $r_0$ for which $f'|_{A_{r_0}^\square}$ is simplicial and computing 
the minimal $r_1$ for which $f'$ is extendable, yields in many cases a~much better estimate of the robustness of zero set than
the estimates based on Lemma~\ref{t:approx}: this can already be seen in the above table where the $r_1$ column is a good approximation
of the true robustness. 
We give the table with smaller grids that continues to higher dimensions.
\begin{center}
\begin{tabular}{c|c|c|c|c}
$n$ & $g$ & min simplicial $r_0$ & $r_1$ & true robustness\\
\hline 
$2$ & $10$ & $0.025$ & $0.889$ & $0.828$\\
$3$ & $10$ & $0.025$ & $=r_0$ & $0$\\
$4$ & $10$ & $0.025$ & $0.667$ & $0.667$\\
$5$ & $10$ & $0.074$ & $=r_0$ & $0$ \\
$6$ & $10$ & $0.074$ & $0.667$ & $0.58$ \\
$7$ & $6$ & $0.241$ & $=r_0$ & $0$ \\
$8$ & $5$ & $0.251$ & $1.0$ & $0.522$
\end{tabular}
\end{center}

It is an interesting question to find natural conditions on functions, other than the relatively weak simplexwise Lipschitz property,
that justify the usage of smaller grids and guarantee that the robustness of zero set is close to the computed minimal $r$ 
for which $f'|_{A_{r}^\square}$ becomes extendable. 

\heading{A function with nontrivial secondary obstruction.}
A second function we experimented with is $h: [-1,1]^{n+1}\to\R^n$ given by
\begin{align*}
& h_1(x)= 2x_0 x_2 + 2x_1 x_3\\
& h_2(x)= 2x_1 x_2 - 2x_0 x_3\\
& h_3(x)= x_0^2+x_1^2-x_2^2-x_3^2\\
& h_4(x)=x_4 \\
& \hskip 1cm \ldots \\
& h_n(x)=x_n.
\end{align*}
The restriction of $h$ to $\partial [-1,1]^{n+1}$ is the generator of the nontrivial homotopy group $\pi_n(S^{n-1})$:
in case $n=3$ it is the Hopf map and for $n>3$ its iterated suspension.
The robustness of zero equals $\min_{x\in \partial [-1,1]^{n+1}} |h(x)|$ and\footnote{The robustness 
is $1$ in the $\ell_2$ norm and $\sqrt{3}-1$ in the max-norm.} 
it is the simplest example where a nontrivial secondary obstruction occurs. Common tests for zero verification such as
the degree test would fail here.

Again, we work with the max-norm $\ell_\infty$ and derive, via elementary calculus, the estimate $2(n+1)$ on the global Lipschitz constant.
This yields $|h(x)-h(y)|\leq \frac{4(n+1)}{g-1}=:\alpha$.
Assume that only $\alpha$ and a list of function values in a grid $g^{n+1}$ is given. If $g$ is large enough so that $r_0:=\alpha$
is smaller than $\min_{\partial} |h(x)|_{\infty}=\sqrt{3}-1$, then the algorithm computes the minimal $r_1>r_0$ for which
$h|_{A_r^\square}$ is extendable to a nowhere zero function $[-1,1]^{n+1}\to\R^{n}\setminus\{0\}$. For this to succeed, we need to take
$g$ at least 22 for $n=3$.

Again, the program gives surprisingly good results for much smaller grids, although Lemma~\ref{t:approx} gives no guarantees. 
In all cases that we tried, we verified that whenever $r_0$ is large enough so that $h'|_{A_{r}^\square}$ is simplicial, 
then it is homotopic to $h/|h|$ and the algorithm computes that the secondary obstruction dies close to the real robustness of zero.
\begin{center}
\begin{tabular}{c|c|c|c|c}
$n$ & $g$ & min simplicial $r_0$ & $r_1$ & true robustness\\
\hline 
$3$ & $10$ & $0.1$ & $0.79$ & $0.732$\\
$4$ & $10$ & $0.111$ & $0.79$ & $0.732$\\
$5$ & $8$ & $0.163$ & $0.816$ & $0.732$\\
\end{tabular}
\end{center}

\end{document}